\documentclass[11pt]{article}
\usepackage{fullpage}
\usepackage{color,xspace}
\usepackage{amsmath,amssymb,amsthm,enumerate}
\usepackage{graphicx}
\usepackage{subcaption}
\usepackage{url}
\usepackage{hyperref}
\hypersetup{colorlinks=true}
\hypersetup{linkcolor=black,anchorcolor=black,citecolor=black,urlcolor=black}
\usepackage[algo2e,ruled]{algorithm2e}
\usepackage{algorithm,  algorithmic}

\newtheorem{theorem}{Theorem}
\newtheorem{definition}[theorem]{Definition}

\newtheorem{proposition}[theorem]{Proposition}

\newtheorem{lemma}[theorem]{Lemma}
\newtheorem{corollary}[theorem]{Corollary}
\newtheorem{claim}[theorem]{Claim}

\newenvironment{proofof}[1]{

\noindent{\bf Proof of {#1}:}}
{\hfill\qed\\

}

\usepackage[colorinlistoftodos,textsize=tiny,textwidth=2cm,color=red!25!white,obeyFinal]{todonotes}

\newcommand{\ignore}[1]{}

\DeclareMathOperator{\OPT}{OPT}

\newcommand{\pathalg}{\textsf{Path-ALG}}

\newcommand{\paths}{\mathcal{P}}
\DeclareMathOperator{\cov}{cov}
\DeclareMathOperator{\width}{width}
\newcommand{\Fhat}{\widehat{F}}

\newcommand{\lambdah}{\widehat{\lambda}}
\newcommand{\lambdaold}{\lambda^{\text{old}}}
\newcommand{\lambdanew}{\lambda^{\text{new}}}
\newcommand{\yold}{y^{\text{old}}}
\newcommand{\ynew}{y^{\text{new}}}
\newcommand{\yhat}{\widehat{y}}
\newcommand{\cmax}{c_{\max}}

\DeclareMathOperator{\class}{class}
\newcommand{\ALG}{\textsf{ALG}}

\newcommand{\R}{\mathcal{R}}

\allowdisplaybreaks

\title{Tight Bounds for Online Weighted Tree Augmentation\thanks{This work was done in part while the authors were visiting the Simons Institute for the Theory of Computing.}}

\author{
Joseph (Seffi) Naor\thanks{Technion, Israel. Email: \href{mailto:naor@cs.technion.ac.il}{\texttt{naor@cs.technion.ac.il}}. Supported in part by ISF grant 1585/15 and BSF grant 2014414.}
\and
Seeun William Umboh\thanks{The University of Sydney, Australia. Email: \href{mailto:william.umboh@sydney.edu.au}{\texttt{william.umboh@sydney.edu.au}}. Supported in part by NWO grant 639.022.211 and ISF grant 1817/17. Part of this work was done while a postdoc at Eindhoven University of Technology, and while visiting the Hebrew University of Jerusalem and the Technion.}
\and
David P. Williamson\thanks{Cornell University, USA. Email: \href{mailto:davidpwilliamson@cornell.edu}{\texttt{davidpwilliamson@cornell.edu}}}
}

\begin{document}

\date{}

\begin{titlepage}
\def\thepage{}
\thispagestyle{empty}
\maketitle

\begin{abstract}
  The Weighted Tree Augmentation problem (WTAP) is a fundamental problem in network design. In this paper, we consider this problem in the online setting. We are given an $n$-vertex spanning tree $T$ and an additional set $L$ of edges (called links) with costs. Then, terminal pairs arrive one-by-one and our task is to maintain a low-cost subset of links $F$ such that every terminal pair that has arrived so far is $2$-edge-connected in $T \cup F$. This online problem was first studied by Gupta, Krishnaswamy and Ravi (SICOMP 2012) who used it as a subroutine for the online survivable network design problem. They gave a deterministic $O(\log^2 n)$-competitive algorithm and showed an $\Omega(\log n)$ lower bound on the competitive ratio of randomized algorithms. The case when $T$ is a path is also interesting:  it is exactly the online interval set cover problem, which also captures as a special case the parking permit problem studied by Meyerson (FOCS 2005). The contribution of this paper is to give tight results for online weighted tree and path augmentation problems.  The main result of this work is a deterministic $O(\log n)$-competitive algorithm for online WTAP, which is tight up to constant factors.

 \end{abstract}
\end{titlepage}

\section{Introduction}
In the {\em weighted tree augmentation problem} (WTAP), we are given an $n$-vertex spanning tree $T = (V,E)$ together with an additional set of edges $L$ called \emph{links}, where $L \subseteq \binom{V}{2}$. Each link $\ell \in L$ has a cost $c(\ell) \geq 0$. Terminal pairs $ (s_i,t_i) $, $i=\{ 1,\ldots,k \}$, are given and the goal is to compute a minimum cost subset of links $F \subseteq L$ such that each terminal pair is (edge) $2$-connected in $T \cup F$. In the unweighted version, the links have unit costs and the problem is known as the {\em tree augmentation problem} (TAP). If the spanning tree $T$ is a path, then the unweighted problem is called the {\em path augmentation problem} (PAP), while the weighted version is called {\em weighted path augmentation} (WPAP).

TAP and WTAP are considered to be fundamental connectivity augmentation problems, and have been studied extensively. TAP is already known to be APX-hard and the best approximation algorithms for WTAP and TAP achieve approximation factors of 2 and 1.458 respectively \cite{FredericksonJ81, GKZ18}. Improving these bounds is an important open problem.

We consider these problems in the online setting. In online WTAP, we are initially given a spanning tree $T = (V,E)$, and the set of links $L$ together with their costs. The terminal pairs $(s_i,t_i)$ arrive online one by one. Our goal is to maintain a low-cost subset of links $F \subseteq L$ such that each terminal pair seen so far is (edge) $2$-connected in $T \cup F$. 

Online WTAP occurs as a subproblem in the online survivable network design algorithm of Gupta, Krishnaswamy and Ravi~\cite{GKR12}. They observed that the online tree augmentation problem can be cast as an instance of the online set cover problem\footnote{In the online set cover problem, elements arrive online and need to be covered upon arrival by sets from a set system known in advance. (Note that not necessarily all elements will appear.)}  in which the elements are the fundamental cuts defined by the terminal pairs and the sets are the links. Since there are only $n$ elements and $O(n^2)$ sets, applying the results of Alon et al.~\cite{AAABN09} yields a fractional $O(\log n)$-competitive algorithm. But, then, how does one round the fractional solution online? Randomized rounding seems to be the only rounding technique we have for this problem, and it yields a randomized $O(\log^2 n)$-competitive algorithm, as observed by \cite{AAABN09}. This competitive factor can even be achieved deterministically at no further cost \cite{AAABN09}. We note that the loss of a logarithmic factor in the rounding step seems inherent. 
Interestingly, Gupta, Krishnaswamy and Ravi~\cite{GKR12} also showed for the rooted setting ($s_i = r$ for some root $r$) a lower bound of $\Omega(\log n)$ against randomized algorithms. It is easy to observe that this lower bound also holds against fractional online algorithms.

There has been a long line of work on maintaining connectivity online, starting in the seminal paper of Imase and Waxman \cite{IW91}. A
$\Theta(\log n)$-competitive algorithm is given there for the online Steiner problem in undirected graphs. In this problem the graph with a fixed root vertex is known in advance and the terminals are given one by one, and one must ensure that all terminals that have arrived so far are connected to the root.
Other polylogarithmic (in $n$)
competitive algorithms have been given for more complex models of connectivity, including those with node costs rather than edge costs and penalties for violating connectivity constraints; see \cite{AAB04, BC97, NaorPS11, HLP13, HLP14, U15, QUW18}.  Gupta, Krishnaswamy, and Ravi \cite{GKR12} consider the online survivable network design problem, which generalizes WTAP. In this problem, a graph is fixed in advance and terminal pairs $(s_i,t_i)$ arrive with connectivity requirements $r_i$; one must ensure that there are at least $r_i$ edge-disjoint paths between $s_i$ and $t_i$ for all pairs that have arrived thus far.  They give a randomized $\tilde{O}(r_{\max} \log^3 n)$-competitive algorithm for the problem, where $r_{\max} = \max_i r_i$. Note that this problem with uniform requirements $r_i = 2$ already generalizes WTAP.

The online WPAP, when $T$ is a path, is an interesting problem in its own right. This problem is equivalent to online interval set cover. It captures as a special case the parking permit problem introduced by Meyerson \cite{M05}.  In this problem, there is a sequence of days; each day it is either sunny or it rains, and if it rains we must purchase a parking permit.  Permits have various durations and costs.  We can model the parking permit problem by online path augmentation by letting the edges of the path correspond to the sequence of days, the links to the permits, and the rainy days to a terminal pair request for the corresponding day.  Meyerson \cite{M05} gives a deterministic $O(\log n)$-competitive algorithm for the problem and a randomized $O(\log \log n)$-competitive algorithm, and shows lower bounds on the competitive ratio of $\Omega(\log n/\log \log n)$ for deterministic algorithms and $\Omega(\log \log n)$ for randomized algorithms.  Note that online WPAP is a strict generalization of the parking permit problem because the parking permit problem assumes that permits of the same duration have the same cost, whereas no such assumption is made of the links in WPAP.

\subsection{Our Results}
The contribution of this paper is to give tight results (within constant factors) for online tree and path augmentation problems. Our main result is that weighted online tree augmentation has a competitive ratio of $\Theta(\log n)$.
\begin{theorem}
  \label{thm:det-tree-ub}
  There is a deterministic algorithm for online WTAP with competitive ratio $O(\log n)$.
\end{theorem}
\noindent This result is tight up to constant factors because of the $\Omega(\log n)$ lower bound on randomized algorithms for WTAP given by \cite{GKR12}.  As we mention above,
\cite{GKR12} gives a randomized $\tilde{O}(r_{\max} \log^3 n)$-competitive algorithm for the online survivable network design problem.  An intriguing open question is whether this competitive ratio can be improved, say to $O(r_{\max} \log n)$ or even $O(\log n)$.  In fact, we are unaware of lower bounds that rule out the latter bound. We view our main result as a necessary stepping stone towards obtaining an $O(r_{\max} \log n)$ or $O(\log n)$ bound. Indeed, for $r_{\max} = 2$,  plugging in our algorithm for online WTAP into the algorithm of \cite{GKR12} improves their competitive ratio from $\tilde{O}(\log^3 n)$ to $\tilde{O}(\log^2 n)$.
\begin{corollary}
  For online survivable network design with $r_{\max} = 2$, there is a randomized algorithm with competitive ratio $\tilde{O}(\log^2 n)$.
\end{corollary}

Our second result shows that the competitive ratio for deterministic algorithms for online path augmentation is also $\Theta(\log n)$.  Meyerson \cite{M05} gives a lower bound of $\Omega(\log n / \log \log n)$ for deterministic algorithms for the parking permit problem, and hence for online path augmentation. We improve the analysis of his lower bound instance to show the following.

\begin{theorem}
  \label{thm:det-path-lb}
  Every deterministic algorithm for online WPAP has competitive ratio $\Omega(\log n)$.
\end{theorem}
\noindent Since we use a parking permit instance to show the lower bound, we have the same lower bound for the parking permit problem.

Finally, we show that the fractional version of online path augmentation has competitive ratio $\Theta(\log \log n)$ for deterministic algorithms. Meyerson \cite{M05} gives a lower bound of $\Omega(\log \log n)$ for randomized algorithms for the parking permit problem, and hence for online fractional path augmentation. Our algorithm implies an exponential gap between the competitive ratios of fractional path augmentation and fractional tree augmentation.  We show the following.

\begin{theorem}
  \label{thm:rand-path-ub}
  There is a deterministic algorithm for online fractional WPAP with competitive ratio $O(\log \log n)$.
\end{theorem}

\noindent Recall that online WPAP is equivalent to online interval set cover. Thus, Theorems~\ref{thm:det-tree-ub} and \ref{thm:rand-path-ub} imply that restricting online set cover to interval sets allows for improved competitive ratios. Also, even though interval set cover and interval hitting set are equivalent in the offline case, the latter turns out to be exponentially more difficult than the former in the online case; in contrast to Theorem~\ref{thm:rand-path-ub}, Even and Smorodinsky~\cite{EvenS14} gave a lower bound of $\Omega(\log n)$ for online fractional hitting set.

\subsection{Our Techniques}
We now outline some of the ideas behind our algorithms.

\paragraph*{Online WTAP} As mentioned  before, there is an online {\em fractional} $O(\log n)$-competitive algorithm for WTAP that follows from the work of \cite{AAABN09} on the online set cover problem. However, it is unclear how to exploit the special structure of the set system in hand in WTAP (as defined by the links) to avoid the loss of another factor of $O(\log n)$ when rounding the fractional solution into an integral one (either randomized or deterministic). Thus, our approach to proving Theorem~\ref{thm:det-tree-ub} takes a completely different route.
There are two key ingredients in our proof:

\begin{enumerate}
\item {\bf Low-width path decomposition.} The first ingredient is a path decomposition of low ``width'': in particular, there is a decomposition of the tree into edge-disjoint paths such that any path in the tree intersects at most $O(\log n)$ paths of the decomposition. Such a decomposition can be obtained using the heavy-path decomposition of Sleator and Tarjan~\cite{SleatorT83}. This immediately implies an $O(\log n)$-approximate black-box reduction from online tree augmentation to online path augmentation. Unfortunately, Theorem~\ref{thm:det-path-lb} gives a lower bound of $\Omega(\log n)$ for the latter problem. Since a tree may have width $\Omega(\log n)$ in the worst case (e.g., a binary tree), the best we can achieve for WTAP using a black-box reduction is a competitive ratio of $O(\log^2 n)$.
\item {\bf Refined guarantee for path augmentation.} The second ingredient is our main technical contribution. We define a notion of \emph{projection} for links onto paths in the path decomposition, and call the projected link \emph{rooted} if it has, as its endpoint, the node of the path closest to the root of the tree.  The key insight is that the path decomposition has a special structure: for each link, its projection is rooted for all but at most one of the paths in the decomposition.  We then give a version of the path algorithm that treats rooted links differently from non-rooted links; in particular, an online path augmentation algorithm that finds a solution whose cost is within a constant factor of the rooted links of the optimal solution plus an $O(\log n)$ factor of the cost of the non-rooted links.  Intuitively, then, summing the cost over all the paths in the decomposition, each link appears as a rooted link in at most $O(\log n)$ paths in the decomposition and as a non-rooted link in at most one path in the decomposition, yielding the $O(\log n)$ factor overall.
\end{enumerate}

\paragraph*{Online Fractional WPAP}
Directly applying the online fractional set cover algorithm of \cite{AAABN09} to online fractional WPAP only yields a competitive ratio of $O(\log n)$. However, for online set cover instances in which each element is covered by at most $d$ sets, the algorithm of \cite{AAABN09} is $O(\log d)$-competitive. Thus, to get a competitive ratio of $O(\log \log n)$, the basic idea is to reduce to a restricted instance in which each request can only be covered by $O(\log n)$ links. For such restricted instances, applying the algorithm of \cite{AAABN09} gives a competitive ratio of $O(\log \log n)$.

\subsection{Other Related Work}
Recently, Dehghani et al. \cite{DehghaniEHLS18} studied online survivable network design, giving a bicriteria approximation algorithm, and considering several stochastic settings.

\subsection{Organization of the Paper}
We start with the preliminaries and describe the low-width path decomposition in Section~\ref{sec:prelim}. In Section~\ref{sec:reduction}, we present the refined guarantee needed for online path augmentation. Then, we show how to achieve the required refined guarantee in Section~\ref{sec:path}.

\section{Preliminaries}
\label{sec:prelim}

We restate the formal definition of the problem. In the online weighted tree augmentation problem, we are initially given a spanning tree $T = (V,E)$, and an additional set of edges called \emph{links} $L \subseteq \binom{V}{2}$ with costs $c(\ell) \geq 0$. Then, terminal pairs $(s_i,t_i)$ arrive one by one. Our goal is to maintain a low-cost subset of links $F \subseteq L$ such that each terminal pair seen so far is $2$-connected in $T \cup F$.

\paragraph*{Notation} Denote by $P(u,v)$ the path between $u$ and $v$ in $T$. For a link $\ell = (u,v)$, we write $P(\ell) = P(u,v)$ and for a set $S$ of links, we write $P(S) = \bigcup_{\ell \in S} P(\ell)$. We say that a link $\ell \in L$ \emph{covers} an edge $e \in E$ if $e \in P(\ell)$. Define $\cov(e) = \{\ell \in L : e \in P(\ell)\}$ to be the set of links covering $e$. Note that $\cov(e)$ is exactly the set of links crossing the cut induced by the tree edge $e$. Let $R \subseteq E$ be a set of requests. Then, a solution $F$ is feasible if and only if for every edge $e \in R$, we have $F \cap \cov(e) \neq \emptyset$; or equivalently, if $P(F) \supseteq R$.

\paragraph*{Simplifying assumptions} In the rest of this paper, we assume that link costs are powers of $2$; this assumption is without loss of generality since we can round up all edge costs and lose only a factor of 2 in the competitive ratio. Given that link costs are powers of $2$, we say that the \emph{class} of a link $\ell$ is $j$ if $c(\ell) = 2^j$ and we write $\class(\ell) = j$.

Given an instance in which link costs are powers of $2$, we also assume that requests are \emph{elementary}: each request  $(s_i,t_i)$ is a tree edge $e \in E$. This is without loss of generality because an adversary can simulate a non-elementary request $(s_i,t_i)$ by a sequence of requests, where each request is an edge along the path between $s_i$ and $t_i$ in $T$.

\paragraph*{Path decomposition}
We next define a rooted path decomposition, see Figure~\ref{fig:decomp} for an example.
\begin{definition}
  [Rooted Path Decomposition]
  Let $T$ be a tree. A \emph{path decomposition} of $T$ is a partition $\paths$ of its edge set into edge-disjoint paths. We say $\paths$ is \emph{rooted} if there is a vertex $r \in T$ such that if we root $T$ at $r$, then for each path $p \in \paths$, the least common ancestor of the vertices on $p$ is an endpoint of the path (we call this endpoint the \emph{root} of $p$). The \emph{width} of $\paths$ is $\width(\paths) = \max_{u,v \in V(T)} |\{p \in \paths: P(u,v) \cap p \neq \emptyset\}|$, the maximum number of paths $p \in \paths$ that any path in $T$ intersects.
\end{definition}

\begin{figure}
  \centering
  \begin{subfigure}[b]{0.4\textwidth}
    \centering
    \includegraphics[scale=0.4]{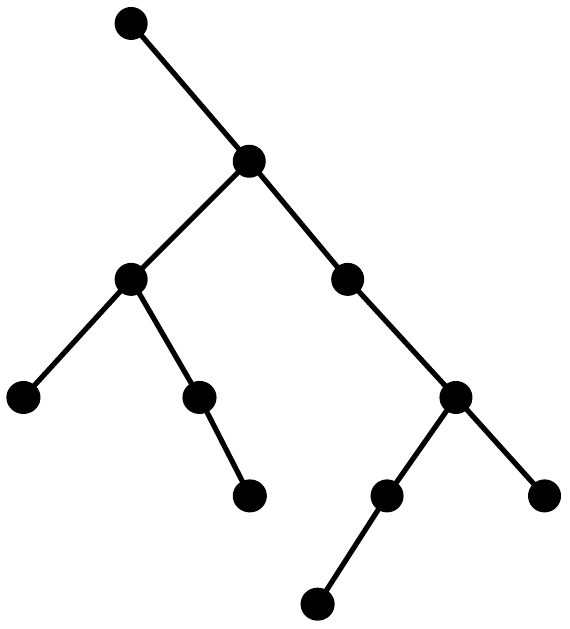}
  \end{subfigure}
  \hfill
  \begin{subfigure}[b]{0.4\textwidth}
    \centering
    \includegraphics[scale=0.4]{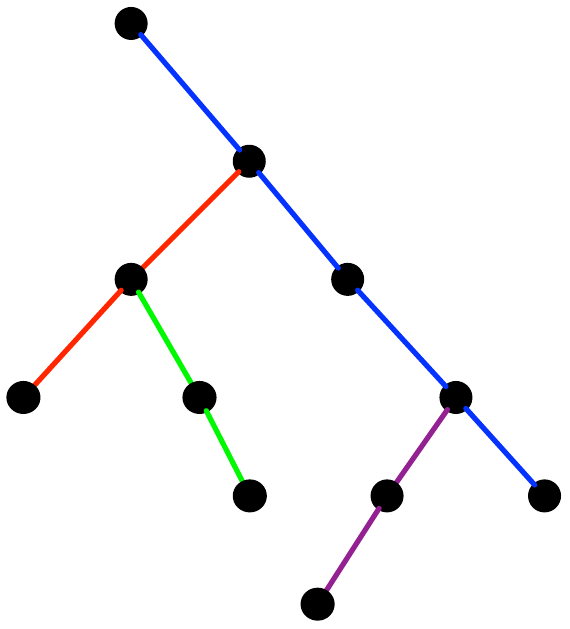}
  \end{subfigure}
  \caption{\small Example of a graph and its rooted path decomposition. The edge colors reflect the partition of the edges. The root of each path is the highest vertex of the path.}
  \label{fig:decomp}
\end{figure}

\begin{lemma}[Existence of Low Width Rooted Path Decompositions]
  \label{lem:pdecomp}
  Every tree on $n$ vertices admits a rooted path decomposition of width $O(\log n)$.
\end{lemma}

An $O(\log n)$-width rooted path decomposition can be obtained using the so-called \emph{heavy path decomposition} of Sleator and Tarjan~\cite{SleatorT83}. For the sake of completeness, we give a proof here. The following notion of a caterpillar decomposition will be convenient.

\begin{definition}
  [Caterpillar Decomposition]
  Let $T$ be a rooted tree on $n$ vertices. A \emph{caterpillar decomposition} of $T$ is a vertex-disjoint decomposition of $T$ into a root-to-leaf path $B$ (called the \emph{backbone}) and subtrees $T_i$ that are connected to $B$. The decomposition is \emph{balanced} if for each subtree $T_i$, we have $|V(T_i)| \leq n/2$.
\end{definition}

\begin{lemma}
  \label{lem:balanced}
  Every tree admits a balanced caterpillar decomposition.
\end{lemma}

\begin{proof}
  The existence of a balanced caterpillar decomposition is an easy consequence of the fact that every tree $T$ has a balanced vertex separator $v$, i.e. after removing $v$ from $T$, each of the remaining connected components has at most $n/2$ vertices. The following is a balanced caterpillar decomposition of $T$: pick an arbitrary root-to-leaf path containing $v$ to be the backbone $B$, and the subtrees $T_i$ to be the connected components of $T$ after removing $B$.
\end{proof}

\begin{proofof}{Lemma~\ref{lem:pdecomp}}
  The lemma easily follows by choosing an arbitrary root vertex of $T$ and recursively applying Lemma \ref{lem:balanced}.
\end{proofof}

\section{Refined Guarantee for Online Path Augmentation}
\label{sec:reduction}

As already mentioned, Lemma~\ref{lem:pdecomp} implies an $O(\log n)$-approximate black-box reduction to online path augmentation: given an $\alpha$-competitive algorithm for online path augmentation, we have an $O(\alpha\log n)$-competitive algorithm for online tree augmentation. However, 
Theorem~\ref{thm:det-path-lb} says that $\alpha = \Omega(\log n)$ for deterministic algorithms. To get around this lower bound, a more refined guarantee for online path augmentation is needed.

We need some notation to describe this refined guarantee. Suppose $\paths$ is a rooted path decomposition of $T$ and $\ell$ a link. For $Q \in \paths$, let $\pi_Q(\ell)$ be the link whose endpoints are endpoints of the path $P(\ell) \cap Q$; we call $\pi_Q(\ell)$ the \emph{projection} of $\ell$ onto $Q$. We say that $\ell$ is \emph{$Q$-rooted} if one of the endpoints of $\pi_Q(\ell)$ is the root of $Q$, and \emph{$Q$-non-rooted} otherwise. (See Fig. ~\ref{fig:projection} for an illustration.) The main ingredient for the refined guarantee is the next lemma.

\begin{figure}
  \centering
    \includegraphics[scale=0.4]{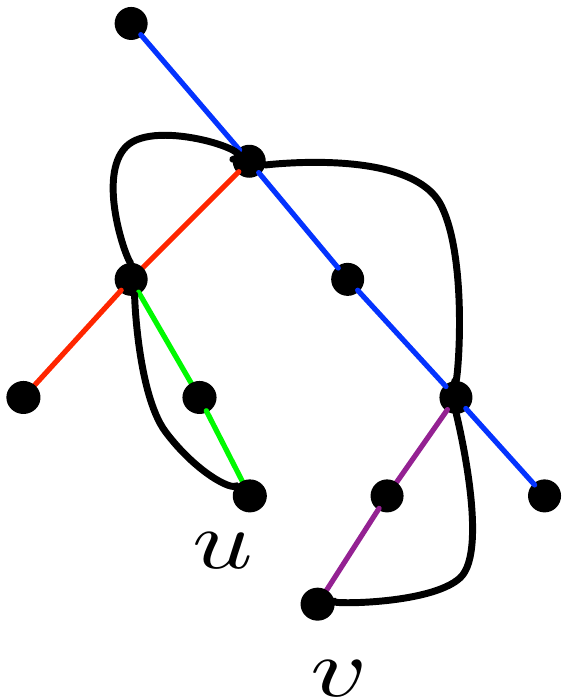}
  \caption{\small Illustration of the projections of the link $(u,v)$ onto the paths of the decompositon. Only the projection onto the blue path is non-rooted.}
  \label{fig:projection}
\end{figure}

\begin{lemma}
  \label{lem:decomp-rooted}
  Consider a tree $T$ and link $\ell = (u,v)$. Suppose $\paths$ is a rooted path decomposition of $T$. Then, there is at most one path $Q \in \paths$ for which $\ell$ is $Q$-non-rooted.
\end{lemma}

\begin{proof}
  We claim that for any path $Q \in \paths$ such that $\ell$ is a non-rooted link, the least common ancestor $a$ of $u$ and $v$ must lie in $Q$ but is not an endpoint of $Q$, i.e. it lies strictly in the middle of $Q$. Since $\paths$ is an edge-disjoint decomposition of $T$, there can be at most one such path and thus the claim implies the lemma.

  We proceed to prove the claim. Let $u',v'$ be the endpoints of $\pi_Q(\ell)$. Since $\paths$ is a rooted path decomposition, either $u'$ is an ancestor of $v'$ or vice versa; suppose the former. We now argue that $u'$ is the least common ancestor of $u$ and $v$. There are two cases: (1) either $u'$ is an endpoint of $\ell$; (2) or there is a vertex $z$ of $P(u,v)$ adjacent to $u'$ but is not on $Q$. In case (1), we are done. Consider case (2). Since $u'$ is not an endpoint of $Q$, its parent must be on $Q$, and thus $z$ is a child of $u'$. Therefore, $u'$ is the least common ancestor of $u$ and $v$.
\end{proof}

Motivated by Lemma~\ref{lem:decomp-rooted}, we define the online rooted path augmentation problem. An instance of online rooted path augmentation consists of a rooted path $Q$ where the root $r$ is an endpoint of $Q$. For such an instance, we say that a link is rooted if one of its endpoints is $r$. Lemma~\ref{lem:decomp-rooted} suggests that we should devise an algorithm for online rooted path augmentation with the following refined guarantee.

\begin{definition}[Nice Solution]
  A solution $F$ for an instance of online rooted path augmentation is \emph{nice} if for any feasible solution $F^*$, we have $c(F) \leq O(1) c(R^*) + O(\log n) c(S^*)$ where $R^*$ is the set of rooted links and $S^*$ is the set of non-rooted links of $F^*$, respectively. An algorithm is nice if it always produces a nice solution.
\end{definition}

\begin{lemma}
  \label{lem:refined}
  Suppose that there exists a deterministic nice algorithm $\pathalg$ for online rooted path augmentation. Then, there exists an $O(\log n)$-competitive deterministic algorithm for online tree augmentation.
\end{lemma}

\begin{proof}
  Here is our algorithm for general instances. Consider a general instance of online weighted tree augmentation with tree $T = (V,E)$, requests $e_1, \ldots, e_k \subseteq E$ and links $L = \binom{V}{2}$ with costs $c(\ell)$. Our algorithm works as follows. By Lemma~\ref{lem:pdecomp}, there exists a rooted path decomposition $\paths$ of $T$ with width $w = O(\log n)$. Now, each rooted path $Q \in \paths$ defines an instance of online rooted path augmentation: the links are $L_Q = \{\pi_Q(\ell) : \ell \in L\}$ where $\pi_Q(\ell)$ has cost $c(\ell)$, and the sequence of requests is exactly the subsequence of requests that lie on $Q$. So, our algorithm runs in parallel $|\paths|$ instantiations of $\pathalg$, one per rooted path $Q \in \paths$. When request $e_i$ arrives, if $e_i \in Q$ (since $e_i$ is an elementary request, it must lie on some path of $\paths$), then our algorithm uses the instantiation of $\pathalg$ on $Q$ to handle that request; in particular, if $\pathalg$ buys the projected link $\pi_Q(\ell)$, then our algorithm buys the link $\ell$.

Let us now analyze the competitive ratio of this algorithm. Let $F^*$ be a feasible solution. 
For $Q \in \paths$, we denote by $R^*_Q$, and $S^*_Q$ the subset of $F^*$ which is $Q$-rooted, and $Q$-non-rooted, respectively. Since $\pathalg$ is nice, we have that our algorithm's solution $F$ has cost
\[c(F) \leq \sum_{Q \in \paths} \left[ O(1) c(R^*_Q) + O(\log n) c(S^*_Q)\right] \leq O(\log n)c(F^*),\]
where the last inequality is because Lemma~\ref{lem:decomp-rooted} implies that each link of $F^*$ is in $S^*_Q$ for at most one $Q \in \paths$ and is in $R^*_Q$ for at most $w = O(\log n)$ paths $Q \in \paths$.
\end{proof}

In the next section, we construct a nice deterministic algorithm. Together with Lemma~\ref{lem:refined} this gives a deterministic $O(\log n)$-competitive algorithm for online tree augmentation, thus proving Theorem~\ref{thm:det-tree-ub}.

\section{A Nice Algorithm for Online Path Augmentation}
\label{sec:path}
In this section, we  devise a nice algorithm for online rooted path augmentation. In the following, we use the convention that the root of the path is the left endpoint of the path.

We begin by showing in Section~\ref{sec:minimal} that it suffices to consider simpler instances that we call \emph{minimal instances}. Then, we describe in Section~\ref{sec:niceness} how to prove niceness using an LP for the problem. Finally, we describe and analyze the algorithm in Sections~\ref{sec:algorithm} and \ref{sec:analysis}.

\subsection{Minimal Instances}
\label{sec:minimal}
The first step is a preprocessing step that simplifies the structure of the link set. In particular, we prune the instance so that it is of the following type.

\begin{definition}[Minimal Instances]
  A set of links $L$ and its costs $c$ are \emph{minimal} if they satisfy the following properties:
  \begin{enumerate}
  \item\label{def:minimal:class} for each class $j$, there is at most one rooted link and for every edge $e$, there are at most two links $\ell$ with $e \in P(\ell)$;
  \item\label{def:minimal:rooted} for any two rooted links $\ell$ and $\ell'$, if $\class(\ell) > \class(\ell')$, then $P(\ell) \supsetneq P(\ell')$.
  \end{enumerate}
  An instance is minimal if its links and costs are minimal.
\end{definition}

Given a set of links $L$ and its costs $c$, we prune $L$ to get a minimal subset of links $L' \subseteq L$ as follows. We begin by pruning the rooted links: while there exists a rooted link $\ell$ and a rooted link $\ell'$ of the same or lower class such that $P(\ell') \supseteq P(\ell)$, remove $\ell$. Then we prune the non-rooted links for each class $j$: let $L_j$ be the set of class $j$ links and $L'_j$ be a minimum-size subset of $L_j$ that covers $L_j$, i.e. $P(L'_j) \supseteq P(L_j)$; then, remove the links $L_j \setminus L'_j$. Such a minimum cover may be computed efficiently using an algorithm for minimum interval cover. By minimality, we have that for any edge $e$, there are at most two links $\ell, \ell' \in L'_j$ such that $e \in P(\ell) \cap P(\ell')$. The following claim shows that any link $\ell \in L_j$ that was pruned away can be replaced with at most three links of $L'_j$ and so restricting to $L'$ only causes the value of the optimal solution to increase by at most a factor of $3$.

\begin{claim}
  \label{clm:containment}
  For every link $\ell \in L_j \setminus L'_j$, there exists (at most) three links $\ell_1, \ell_2, \ell_3 \in L'_j$ with $P(\ell) \subseteq P(\ell_1) \cup P(\ell_2) \cup P(\ell_3)$.
\end{claim}

\begin{proof}
  Suppose, towards a contradiction, that there exists a link $\ell \in L_j \setminus L'_j$ whose path $P(\ell)$ can only be covered by at least four links of $L'_j$. Let $S \subseteq L'_j$ be a minimum-size cover of $P(\ell)$ and let $\ell_1$ be the link of $S$ whose left endpoint is leftmost and $\ell_2$ be the link of $S$ whose right endpoint is rightmost. Since $S$ is minimum-size, we have that every link $\ell' \in S \setminus \{\ell_1, \ell_2\}$ has its path $P(\ell') \subseteq P(\ell)$ (otherwise we can discard either $\ell_1$ or $\ell_2$ from $S$ to get a smaller cover of $P(\ell)$). Therefore, swapping $S \setminus \{\ell_1, \ell_2\}$ with $\ell$ gives a smaller cover of $P(L_j)$. However, this contradicts the fact that $L'_j$ is a minimum cover of $P(L_j)$.
\end{proof}

Given a subset of links $L' \subseteq L$,  we say that a solution $F \subseteq L'$ is \emph{nice for $L'$} if for any feasible solution $F' \subseteq L'$, we have $c(F) \leq O(1) c(R') + O(\log n) c(S')$ where $R'$ is the set of rooted links and $S'$ is the set of non-rooted links of $F'$, respectively. The following lemma says that it suffices to have a solution that is nice for a pruning of $L$ and thus it suffices to devise a nice algorithm for minimal instances. 

\begin{lemma}
  \label{lem:minimal}
  Let $L' \subseteq L$ be a pruning of $L$. Then, a solution that is nice for $L'$ is also nice for $L$.
\end{lemma}

\begin{proof}
  Suppose $F^* \subseteq L$ is a feasible solution. Let $R^* \subseteq F^*$ be its set of rooted links and $S^* \subseteq F^*$ be its set of non-rooted links. We now find a rooted link $\ell_r \in L'$ and a set of non-rooted links $S' \subseteq L'$ such that $F' = S' \cup \{\ell_r\}$ is feasible and moreover, $c(\ell_r) \leq O(1)c(R^*)$ and $c(S') \leq O(1) c(S^*)$. Let $\ell \in R^*$ be the link in $R^*$ whose right endpoint is furthest from the root $r$. If $\ell \in L'$, then we set $\ell_r = \ell$. Otherwise, by construction, there exists a rooted link $\ell' \in L'$ of no higher class than $\ell$ and $P(\ell') \supseteq P(\ell)$; then we set $\ell_r = \ell'$. Observe that in both cases, $P(\ell_r) \supseteq P(R^*)$ and that $c(\ell_r) \leq c(R^*)$. We now proceed with the construction of $S'$. We construct $S'$ iteratively, starting with $S' = \emptyset$. Then, for each class $j$, we iterate over each class-$j$ link $\ell \in S^*$: if $\ell \in L'_j$, then we add $\ell$ to $S'$ as well; otherwise, we add to $S'$ the (at most) three links guaranteed by Claim~\ref{clm:containment} $\ell_1, \ell_2, \ell_3 \in L'_j$ with $P(\ell) \subseteq P(\ell_1) \cup P(\ell_2) \cup P(\ell_3)$. Observe that $P(S') \supseteq P(S^*)$ and $c(S') \leq 3c(S^*)$. Since $P(S' \cup \{\ell_r\}) \supseteq P(F^*)$, we get that $F'$ is feasible.

Finally, a solution $F$ that is nice for $L'$ satisfies
\[c(F) \leq O(1)c(\{\ell_r\}) + O(\log n) c(S') \leq O(1)c(R^*) + O(\log n)c(S^*),\]
and so it is also nice for $L$.
\end{proof}

Henceforth, we will focus on devising a nice algorithm for minimal instances.

\subsection{Proving Niceness via the Dual LP}
\label{sec:niceness}
Our algorithm uses the standard LP formulation of the problem. Let $\R$ be the set of requests. The following are the primal and dual LPs, respectively.

\begin{equation}
\boxed{
  \label{lp:rP}
\begin{aligned}
  \mbox{minimize} \quad
  & \sum_{\ell \in L} x(\ell)c(\ell)\\
  \mbox{subject to}\quad
  & \sum_{\ell \in \cov(e)} x(\ell) \geq 1 &\quad \forall e \in \R
\end{aligned}
}
\end{equation}
\begin{equation}
\boxed{
  \label{lp:rD}
\begin{aligned}
  \mbox{maximize} \quad
  & \sum_{e \in \R} y(e)\\
  \mbox{subject to}\quad
  & \sum_{e \in P(\ell)} y(e) \leq c(\ell) &\quad \forall \ell \in L
\end{aligned}
}
\end{equation}
We say that a link $\ell$ is \emph{tight} with respect to a dual solution $y$ if $\sum_{e \in P(\ell)} y(e) = c(\ell)$.

The following lemma tells us how to use the dual to prove niceness.
\begin{lemma}
  \label{lem:dual-nice}
  Let $F$ be a solution. Suppose $y$ is a dual solution such that
  \begin{enumerate}
  \item $c(F) \leq O(1) \sum_e y(e)$,
  \item $\sum_{e \in P(\ell)} y(e) \leq O(\log n) c(\ell)$ for every non-rooted link $\ell$, and
  \item $\sum_{e \in P(\ell)} y(e) \leq O(1) c(\ell)$ for every rooted link $\ell$.
  \end{enumerate}
  Then, $F$ is a nice solution.
\end{lemma}

\begin{proof}
  Let $F^*$ be a feasible solution, $R^*$ be the subset of $F^*$ that is rooted and $S^*$  the subset that is non-rooted. We now show that $\sum_e y(e) \leq O(1) c(R^*) + O(\log n) c(S^*)$, which then implies that $c(F) \leq O(1) c(R^*) + O(\log n) c(S^*)$. Since we have a dual variable $y(e)$ for each request $e$ and $F^*$ is feasible, we have that
  \[\sum_e y(e) \leq \sum_{e \in P(R^*)}y(e) + \sum_{e \in P(S^*)}y(e).\] Using the fact that $\sum_{e \in P(\ell)} y(e) \leq O(1) c(\ell)$ for every rooted link $\ell$, we also have
  \[\sum_{e \in P(R^*)}y(e) \leq \sum_{\ell \in R^*}\sum_{e \in P(\ell)} y(e) \leq O(1)c(R^*).\]
  Similarly, we get that $\sum_{e \in P(S^*)}y(e) \leq O(\log n) c(S^*)$. Putting all of these together, we conclude that $\sum_e y(e) \leq O(1) c(R^*) + O(\log n) c(S^*)$, as desired.
\end{proof}

\subsection{Algorithm}
\label{sec:algorithm}
We now give some of the ideas behind our algorithm.

\paragraph*{An $O(\log n)$-competitive algorithm} 
First, we describe a simple algorithm that constructs a solution $F$ and a dual solution $y$ that satisfies $c(F) \leq O(1) \sum_e y(e)$ and $\sum_{e \in P(\ell)}y(e) \leq O(\log n) c(\ell)$ for \emph{every} link $\ell$. The algorithm maintains a maximal feasible dual solution $y$ and is as follows: when a request $e_i$ arrives, raise its dual variable $y(e_i)$ until some link $\ell$ with $e_i \in P(\ell)$ goes tight; add this link to $F$. There are two parts to the analysis. First, let $\Fhat$ be the set of links in $F$ that cost at least $\max_{\ell \in F}c(\ell)/n^2$. Since $|F| \leq n^2$, we get that $c(F) \leq 2c(\Fhat)$ so it suffices to bound $c(\Fhat)$. The second part of the analysis uses the following charging argument to bound $c(\Fhat)$: whenever we add a tight link $\ell$ to $\Fhat$, we charge its cost to the dual variables $y(e)$ for $e \in P(\ell)$. Let $\lambda(e)$ be the total number of links charged to $y(e)$ and $\yhat$ be the dual solution where $\yhat(e) = \lambda(e)y(e)$. We have $c(\Fhat) \leq O(1) \sum_e \lambda(e)y(e)$. Now observe that $\lambda(e) \leq O(\log n)$ because Property~\ref{def:minimal:class} of minimal instance implies that there can be at most 2 links  $\ell \in \Fhat$ with $e \in P(\ell)$ for a single cost class, and, by definition, $\Fhat$ can have at most $O(\log n)$ cost classes.  So, for each link $\ell$, we have
\[\sum_{e \in P(\ell)}\lambda(e)y(e) \leq O(\log n)\sum_{e \in P(\ell)}y(e) \leq O(\log n)c(\ell)\]
where the last inequality follows from the fact that $y$ is feasible.

\paragraph*{Saving the rooted links} A natural idea to ensure that $\sum_{e \in P(\ell)}\lambda(e)y(e) \leq O(1) c(\ell)$ for each rooted link $\ell$ is to modify the above algorithm to explicitly take into account the charging method as follows: after buying the tight link (we call this a type-1 link), if there is a rooted link $\ell'$ such that $\sum_{e \in P(\ell')}\lambda(e)y(e) > c(\ell')$, buy the one of highest class among such links (we call this a type-2 link). Moreover, we also modify the charging method to only charge each type-1 link $\ell$ to the dual variables $y(e)$ for $e \notin P(\ell')$ where $\ell'$ is the last type-2 link bought.

As we will see later, these  modifications allow us to argue that $\sum_{e \in P(\ell)}\lambda(e)y(e) \leq O(1) c(\ell)$ for each rooted link $\ell$. However, it also introduces a complication: it might be possible that for some type-1 link $\ell$, most of the dual variables $y(e)$ paying towards its cost have $e \in P(\ell')$ where $\ell'$ is the last type-2 link bought. Since the charging method only charges to dual variables $y(e)$ for $e \notin P(\ell')$, this would mean that it might charge an amount that is much less than the cost of $\ell'$.

\paragraph*{Fixing the complication} To fix the above issue, whenever we buy a type-2 link $\ell'$, we also buy all links $\ell''$ of class at most $\class(\ell')$ that crosses $\ell'$, i.e. $\emptyset \subsetneq P(\ell'') \cap P(\ell') \subsetneq P(\ell')$. Property~\ref{def:minimal:class} implies that the total cost of these links is at most $O(1)c(\ell')$. We call these links type-3 links. This ensures that later on, when we buy a type-1 link $\ell$, if $P(\ell) \cap P(\ell') \neq \emptyset$, then $\ell$ must be of higher class than $\ell'$ and thus most of its cost is paid for by dual variables $y(e)$ for $e \notin P(\ell')$.

We describe the complete algorithm formally in Algorithm~\ref{alg:path}. In Algorithm~\ref{alg:path}, we use $Z$ to keep track of $P(\ell)$ where $\ell$ is the last type-2 link bought so far ($Z = \emptyset$ if no type-2 link is bought yet). The links bought in Step~\ref{alg:type1}, \ref{alg:type2}, \ref{alg:type3}, are type-1, type-2, and type-3 links, respectively.

\subsection{Analysis of Algorithm}
\label{sec:analysis}
We now prove that Algorithm~\ref{alg:path} is nice. Let $F_1, F_2, F_3 \subseteq F$ be the sets of type-1, type-2 and type-3 links, respectively. The proof consists of three steps. 
First, we show that $c(F) \leq O(1)c(F_1)$ (Lemma~\ref{lem:alg-cost}) and thus it suffices to bound the cost of type-1 links. Then, we construct a dual solution $\yhat$ that accounts for the cost of type-1 links (Lemma~\ref{lem:alg-yhat}). This shows that $\yhat$ satisfies the first condition of Lemma~\ref{lem:dual-nice}. 
Finally, Lemmas~\ref{lem:rooted} and \ref{lem:non-rooted} show that $\yhat$ satisfies the remaining conditions of Lemma~\ref{lem:dual-nice}.

\begin{algorithm}
\caption{Nice algorithm for online rooted path augmentation}
\begin{algorithmic}[1]
  \label{alg:path}
  \STATE $F \leftarrow \emptyset; y \leftarrow 0; \lambda \leftarrow 0; Z \leftarrow \emptyset$
  \FOR {each unsatisfied request $e_i$}
  \STATE Increase $y(e_i)$ until some link $\ell$ with $e_i \in P(\ell)$ goes tight
  \STATE\label{alg:type1} Add such a link $\ell$ to $F$
  \FOR {each $e \in P(\ell) \setminus Z$ such that $y(e) > 0$}
  \STATE $\lambda(e) \leftarrow \lambda(e) + 1$
  \ENDFOR
  \IF {there exists a rooted link $\ell \notin F$ such that $\sum_{e \in P(\ell)} \lambda(e)y(e) \geq c(\ell)$}
  \STATE\label{alg:type2} Among such links, add to $F$ the link $\ell$ of highest class
  \FOR {$j \leq \class(\ell)$}
  \STATE\label{alg:type3} Add to $F$ all class-$j$ links $\ell'$ that cross $\ell$, i.e. $\emptyset \subsetneq P(\ell') \cap P(\ell) \subsetneq P(\ell)$
  \ENDFOR
  \STATE $Z \leftarrow P(\ell)$
  \ENDIF
  \ENDFOR
\end{algorithmic}
\end{algorithm}

For each type-$1$ link $\ell \in F_1$, define $C(\ell)$ to be the set of edges $e$ such that $\lambda(e)$ was incremented during the iteration that $\ell$ was assigned to $F_1$, i.e. each dual variable $y(e)$ for $e \in C(\ell)$ contributes towards paying $c(\ell)$. Observe that $\lambda(e) = |\{\ell : e \in C(\ell)\}|$ and $C(\ell) \subseteq P(\ell)$.

\begin{proposition}
  \label{prop:alg}
  Algorithm~\ref{alg:path} satisfies the following properties. Let $Z_i$ and $\lambda_i$ denote $Z$ and $\lambda$ at the end of the $i$-th iteration. Then, for every iteration $i$, we have
  \begin{enumerate}
  \item\label{prop:alg:Z} $Z_i \supseteq Z_{i-1}$;
  \item\label{prop:alg:dual} if $y(e_i) > 0$, then $\lambda_i(e_i) > 0$.
  \end{enumerate}
\end{proposition}

\begin{proof}
  The first follows from Property~\ref{def:minimal:rooted} of minimal instances. The second follows from the fact that in the iteration that $e_i$ arrives, since it is unsatisfied, it must not be contained in $Z$. Let $\ell$ be the link added to $F$ in that iteration. Since $e_i \in P(\ell) \setminus Z$ and $y(e_i) > 0$, we have that $\lambda(e_i)$ is increased by $1$ during the iteration and thus $\lambda_i(e_i) > 0$.
\end{proof}

\begin{lemma}
  \label{lem:alg-cost}
  $c(F) \leq O(1) c(F_1)$.
\end{lemma}

\begin{proof}
  We will show that $c(F_3) \leq O(1)c(F_2)$, that $c(F_2) \leq O(1)\sum_e \lambda(e)y(e)$ and that $\sum_e \lambda(e)y(e) \leq c(F_1)$.   Let $\ell_r$ be the last type-2 link bought. We have that $c(\ell_r) \leq \sum_{e \in P(\ell_r)} \lambda(e)y(e)$ by construction. Moreover, since $c(\ell_r) \geq c(\ell)$ for every $\ell \in F_2$ and there is at most one rooted link of each class, we get that $c(F_2) \leq 2 c(\ell_r)$. Thus, we get that $c(F_2) \leq 2 \sum_{e \in P(\ell_r)} \lambda(e)y(e)$. For each type-2 link $\ell$ bought, we buy at most two type-3 links per class $j \leq \class(\ell)$ because of Property~\ref{def:minimal:class} of minimal instances. Therefore, we have $c(F_3) \leq 2 c(F_2) \leq 4 \sum_{e \in P(\ell_r)} \lambda(e)y(e)$.

  Finally, we show that $\sum_e \lambda(e)y(e) \leq c(F_1)$. Since $\lambda(e) = |\{\ell : e \in C(\ell)\}|$, we have \[\sum_e \lambda(e) y(e) = \sum_{\ell \in F_1} \left(\sum_{e \in C(\ell)} y(e)\right).\]
Now, since $C(\ell) \subseteq P(\ell)$ and $y$ is feasible, we get
\[\sum_{e \in C(\ell)}y(e) \leq \sum_{e \in P(\ell)}y(e) \leq c(\ell).\] 
Combining the previous two inequalities gives us that $\sum_e \lambda(e)y(e) \leq c(F_1)$.
\end{proof}

Let $\cmax = \max_{\ell \in F_1} c(\ell)$. Define $\Fhat_1 = \{\ell \in F_1 : c(\ell) \geq \cmax/n^2\}$ and $\lambdah(e) = |\{\ell \in \Fhat_1 : e \in C(\ell)\}|$. We now show that $F$ and the dual solution $\yhat$ where $\yhat(e) = \lambdah(e)y(e)$ satisfies the conditions of Lemma~\ref{lem:dual-nice}.

\begin{lemma}
  \label{lem:alg-yhat}
  $c(F_1) \leq O(1) \sum_e \lambdah(e)y(e)$.
\end{lemma}

\begin{proof}
  Observe that $c(F_1) \leq 2c(\Fhat_1)$ so it suffices to prove that
  \begin{equation}
    c(\Fhat_1) \leq O(1) \sum_e \lambdah(e)y(e).\label{eq:Fhat1}
  \end{equation}
  We now show that this inequality holds at the end of each iteration of the algorithm. Consider an iteration in which the current request $e_i$ is not already covered and suppose $\ell \in \Fhat_1$ is the type-1 link bought in this iteration. The LHS of Inequality~\eqref{eq:Fhat1} increases by $c(\ell)$ in this iteration. We now show that $\sum_e \lambdah(e)y(e)$ increases by at least $c(\ell)/2$. In this iteration, $\lambdah(e)$ increases by $1$ for every $e \in P(\ell) \setminus Z$ and $y(e) > 0$, and so $\sum_e \lambdah(e)y(e)$ increases by exactly $\sum_{e \in P(\ell) \setminus Z} y(e)$.

In the remainder of the proof, we show that $\sum_{e \in P(\ell) \setminus Z} y(e) \geq c(\ell)/2$. If $P(\ell) \cap Z = \emptyset$, then $\sum_{e \in P(\ell) \setminus Z} y(e) = \sum_{e \in P(\ell)} y(e) = c(\ell)$ since $\ell$ is tight. Now suppose $P(\ell) \cap Z \neq \emptyset$. Let $\ell'$ be the type-2 link such that $Z = P(\ell')$. Since $P(\ell) \cap P(\ell') \neq \emptyset$, it must be the case that $\ell$ is of type higher than $\class(\ell')$. This is because otherwise, $\ell$ would have been bought earlier as a type-3 link in the same iteration as $\ell'$. But then since $e_i \in P(\ell)$, it would contradict the assumption that $e_i$ is not already covered at the start of the current iteration. Thus, $\class(\ell) > \class(\ell')$ and so $c(\ell) \geq 2c(\ell')$. So, we now have
\[\sum_{e \in P(\ell) \setminus Z}y(e) \geq \sum_{e \in P(\ell)} y(e) - \sum_{e \in P(\ell')} y(e) \geq c(\ell) - c(\ell') \geq c(\ell)/2,\]
where the second last inequality follows from the fact that $y$ is a feasible dual and $\ell$ is tight. Therefore, Inequality~\eqref{eq:Fhat1} holds at the end of each iteration, as desired.
\end{proof}

Lemmas~\ref{lem:alg-cost} and \ref{lem:alg-yhat} imply that $c(F) \leq O(1) \sum_e \yhat(e)$.

\begin{lemma}
  \label{lem:non-rooted}
  For each non-rooted link $\ell$, we have $\sum_{e \in P(\ell)} \lambdah(e)y(e) \leq O(\log n) c(\ell)$.
\end{lemma}

\begin{proof}
  Property~\ref{def:minimal:class} of minimal instances implies that for each $j$, there are at most two links $\ell' \in \Fhat_1$ of class $j$ with $e \in C(\ell')$. Since each link in $\Fhat_1$ has cost between $\cmax/n^2$ and $\cmax$ and link costs are powers of 2, we have that $\lambdah(e) \leq O(\log n)$. Thus we get that $\sum_{e \in P(\ell)}\lambdah(e)y(e) \leq O(\log n)\sum_{e \in P(\ell)} y(e) \leq O(\log n)c(\ell)$, where the last inequality follows from the fact that $y$ is a feasible dual.
\end{proof}

\begin{lemma}
  \label{lem:rooted}
  For each rooted link $\ell$, we have $\sum_{e \in P(\ell)} \lambdah(e)y(e) \leq O(1) c(\ell)$.
\end{lemma}

\begin{proof}
  We will in fact show that $\sum_{e \in P(\ell)}\lambda(e)y(e) \leq O(1) c(\ell)$. Suppose, at the end of some iteration, that we have $\sum_{e \in P(\ell)} \lambda(e)y(e) > c(\ell)$. Consider the earliest iteration that this happens. We now show that $\sum_{e \in P(\ell)}\lambda(e)y(e) \leq O(1) c(\ell)$ at the end of the iteration and later show that the LHS cannot increase in future iterations.

Let $\lambdaold(e)$ and $\yold(e)$ denote the values of $\lambda(e)$ and $y(e)$ at the start of the iteration and $\lambdanew(e)$ and $\ynew(e)$ denote their values at the end. We have that \[\sum_{e \in P(\ell)} \lambdaold(e)\yold(e) < c(\ell).\] We now show that \[\sum_{e \in P(\ell)} \lambdanew(e)\ynew(e) \leq 3c(\ell).\] Let $e_i$ be the request of the current iteration. During this iteration, we only increase $y(e)$ for $e = e_i$ and we set $\lambda(e_i) = 1$ so $\lambdanew(e_i)\ynew(e_i) = y(e_i)$. So, we have
\[\sum_{e \in P(\ell)} \lambdanew(e)\ynew(e) = \sum_{e \in P(\ell) \setminus \{e_i\}} \lambdanew(e)\yold(e) + y(e_i).\]
Since $y$ is a feasible dual, we have that $y(e_i) \leq c(\ell)$.
Now, Proposition~\ref{prop:alg} implies that $\lambdaold(e) \geq 1$ if $\yold(e) > 0$. Together with the fact that $\lambdanew(e) \leq \lambdaold(e) + 1$, we get that $\lambdanew(e)\yold(e) \leq 2\lambdaold(e)\yold(e)$ and so
\[\sum_{e \in P(\ell) \setminus \{e_i\}} \lambdanew(e)\ynew(e) \leq 2 \sum_{e \in P(\ell) \setminus \{e_i\}} \lambdaold(e)\yold(e) < 2c(\ell).\]
Thus, $\sum_{e \in P(\ell)}\lambdanew(e)\ynew(e) \leq 3c(\ell)$ at the end of the current iteration.

  Finally, we show that $\sum_{e \in P(\ell)}\lambda(e)y(e)$ does not increase in future iterations. At the end of the current iteration, $\ell$ is a candidate to be added to $F$. Among all candidates, the one with highest class is added, so either $\ell$ is added to $F$ or a rooted link $\ell'$ of higher class is added to $F$. In the second case, by Proposition~\ref{prop:alg}, we have $P(\ell') \supseteq P(\ell)$. Thus, in either case, we have  that $Z \supseteq P(\ell)$ at the end of the current iteration. Moreover, in future iterations, we still have $Z \supseteq P(\ell)$ by Proposition~\ref{prop:alg}. Therefore, $\sum_{e \in P(\ell)}\lambda(e)y(e)$ does not increase in future iterations. Thus, we conclude that $\sum_{e \in P(\ell)}\lambda(e)y(e) \leq 3 c(\ell)$ at the end of the algorithm.
\end{proof}

Therefore, we conclude that Algorithm~\ref{alg:path} is nice. Together with Lemma~\ref{lem:refined}, we get Theorem~\ref{thm:det-tree-ub}.

\section{Lower Bound for Online Path Augmentation}
\label{sec:lb}
In this section, we prove Theorem~\ref{thm:det-path-lb}. 

Let $B$ be a constant to be fixed later. Consider the path with $n$ edges and $n+1$ vertices where $n$ is a power of $2B$. In the following, for $i \in \{0, \ldots, n\}$, we write $v_i$ as the $i$-th vertex of the path with $v_0$ being the leftmost vertex. We have $\log_{2B} n$ classes of links. Each link $\ell$ of class $j$ has cost $B^j$ and path length $|P(\ell)| = (2B)^j$; moreover, the class $j$ links are disjoint, i.e. for any $\ell, \ell'$ of class $j$, we have $P(\ell) \cap P(\ell') = \emptyset$, and they cover the entire path. In particular, each link $\ell$ of class $j$ is of the form $\ell = (i\cdot (2B)^j, (i+1)\cdot (2B)^j)$ for some $i \in \left\{0, \ldots, \frac{n}{(2B)^j}-1\right\}$. We say that a link $\ell$ \emph{contains} another link $\ell'$ if $P(\ell) \supseteq P(\ell')$. Note that the links form a hierarchical structure: a link $\ell$ of class $j$ contains $2B$ links of class $j-1$; we call the latter \emph{child links} of $\ell$, and $\ell$ their \emph{parent link}. 
We also say that the set of minimal links (i.e. those at the bottom of the hierarchical structure) are \emph{leaf links}. Note that these correspond exactly to the edges of the path.

We say that an algorithm is \emph{canonical} if for each request $e$, if it buys a class $j$ link $\ell$, then for each class $j' < j$, it also buys the unique link $\ell'$ which has $e \in P(\ell') \subseteq P(\ell)$. Since the costs are geometric, it is easy to see that we can make any algorithm canonical and lose only a constant factor in the competitive ratio. Thus, it suffices to prove a lower bound against canonical algorithms.

\begin{lemma}
  For every $B \geq 2$, every canonical algorithm $\ALG$ has competitive ratio at least $\Omega(\log_B n)$.
\end{lemma}

\begin{proof}
  Consider the following sequence of requests: while there exists an edge that is not yet covered by $\ALG$, the adversary gives as the next request the leftmost edge that is not yet covered. Let $R$ be the sequence of requests in this instance and $F$ be the set of links bought by $\ALG$. The plan is to show that for each class $j$, there exists a feasible solution $F'_j$ consisting only of class $j$ links and that $\sum_j c(F'_j) \leq 2c(F)$. This would imply that $\OPT \leq 2c(F)/\log_{2B}n$ and thus imply the lemma. 

For each class $j$, define the set $F'_j$ to be the links $\ell$ of class $j$ such that there exists a request that is contained in $P(\ell)$. Let $F' = \bigcup_j F'_j$. Clearly, each $F'_j$ is a feasible solution. It remains to show that $c(F') \leq 2c(F)$. To do this, we now prove that $c(F' \setminus F) \leq c(F)$. The following claim will be useful. It says that if a link $\ell$ has $P(\ell)$ containing some request $e$ and $\ell$ was not chosen by the algorithm, then it must be the case that each child link $\ell'$ of $\ell$ has $P(\ell')$ containing some request $e'$.

\begin{claim}
  For each non-leaf link $\ell \in F' \setminus F$, every child link $\ell'$ of $\ell$ is in $F'$ as well.
\end{claim}

\begin{proof}
  We need to show that  for each child link $\ell'$ of $\ell$, there exists a request $e \in P(\ell')$. Since $\ell \notin F$ and the algorithm is canonical, it must be the case that for each request $e \in R \cap P(\ell)$, when the request $e$ arrived, the algorithm covered it using a link that is strictly contained in $\ell$. Thus, by construction of the request sequence, we have that for each child link $\ell'$ of $\ell$, there exists a request contained in $P(\ell')$ and so by definition of $F'$, we have that $\ell' \in F'$.
\end{proof}

We use a token-based argument to prove that $c(F' \setminus F) \leq c(F)$. There will be two types of tokens: \emph{original tokens} and \emph{virtual tokens}. Each token $a$ has a value $v(a)$; moreover, if $a$ is virtual, then it is also associated with a \emph{parent token} $p(a)$ (which may be either an original or virtual token) and we say that $a$ is a \emph{child token} of $p(a)$. Initially, there are no tokens. For each link $\ell \in F$, we give it an original token of value $c(\ell)$. We will now create virtual tokens to pay for links in $F' \setminus F$. We say that a link $\ell \in F' \setminus F$ is paid for if it has tokens of value at least $c(\ell)$. We create virtual tokens iteratively in a bottom-up manner: While there exists a link in $F' \setminus F$ that is not paid for, let $\ell$ be the link of smallest class in $F' \setminus F$ that is not paid for; then, for each token $a$ associated with a child link of $\ell$, we create and give to $\ell$ a virtual token $a'$ with $v(a') = v(a)/2$ and $p(a') = a$. We now argue that in each iteration of this procedure, $\ell$ is paid for. First, observe that $\ell$ cannot be a leaf link (and thus the above procedure is actually valid).  By definition of $F'$, the links in $F'$ that are leaf links correspond exactly to the set of requested edges. Since the algorithm is canonical, these links are also in $F$ and are thus paid for. Therefore, $\ell$ is not a leaf link; moreover, the above claim implies that each of its $2B$ child links are in $F'$. Since $\ell$ is the link of smallest class that is not paid for, its child links are already paid for so they have tokens of total value at least $2B\cdot B^{j-1} = 2B^j$. Therefore, we get that the tokens given to $\ell$ have value at least $B^j$ and so $\ell$ is now paid for.

Since each link of $F' \setminus F$ is paid for by virtual tokens, we get that $c(F' \setminus F)$ is at most the value of virtual tokens. We now argue that the total value of the virtual tokens is at most the value of the original tokens. Note that each token, virtual or original, can have at most one child token. Moreover, every virtual token is a descendant of an original token. Finally, the value of a virtual token is half its parent token's value. Putting the above together, we get that for each original token of value $B^j$, the total value of its descendant virtual tokens is at most $\sum_{i \geq 1} B^j/2^i \leq B^j$. Therefore, the total value of virtual tokens is at most the total value of original tokens. Since the value of the original tokens is exactly $c(F)$, we get that $c(F' \setminus F) \leq c(F)$, as desired. This completes the proof of the lemma.
\end{proof}

Using $B = 2$ in the above lemma gives us Theorem~\ref{thm:det-path-lb}.

\section{A Fractional Algorithm for Online Path Augmentation}
\label{sec:frac}
In this section, we prove Theorem~\ref{thm:rand-path-ub}. We begin by defining the online fractional path augmentation problem formally. We have variables $x(\ell)$ for each link $\ell$ which represent the fraction of $\ell$ that we have bought so far. Initially, $x(\ell) = 0$ for every $\ell$. Then, when a request $e_i$ arrives, we need to increase variables so that $\sum_{\ell : e_i \in P(\ell)} x(\ell) \geq 1$. Moreover, variables cannot be decreased. The cost of a fractional solution $x$ is $\sum_\ell x(\ell) c(\ell)$. The goal is to maintain a feasible fractional solution with low cost. Equivalently, we can think of the online fractional problem in terms of solving LP~\eqref{lp:rP} online: the constraints of the LP arrive one by one, and we need to maintain a feasible fractional solution.

Throughout this section, we assume that we are working with minimal instances as defined in Section~\ref{sec:path}. The following lemma is crucial: it implies that to be competitive against any fractional solution, it suffices to be competitive against any integral solution.

\begin{lemma}[Integrality of LP~\eqref{lp:rP}]
  \label{lem:LP-integral}
  For any instance of path augmentation, there is an optimal solution to LP~\eqref{lp:rP} that is integral.
\end{lemma}

\begin{proof}
  The rows of the constraint matrix of LP~\eqref{lp:rP} satisfy the consecutive-ones property and thus is totally unimodular.
\end{proof}
Henceforth, we will use $\OPT$ to denote the value of an optimal integral solution and our goal is to be competitive against $\OPT$.

Let us first sketch a simple algorithm assuming that we are given the value of $\OPT$ at the beginning. Since $\OPT$ is the value of an optimal \emph{integral} solution, we can remove all links of cost more than $\OPT$ without loss of generality. Then for each request $e_i$, if $e_i$ can be covered by a link $\ell$ of cost at most $\OPT/n$, then we set the variable $x(\ell) = 1$; otherwise, we perform multiplicative update on the variables $\{x(\ell) : e_i \in P(\ell)\}$ \`{a} la the online fractional set cover algorithm of \cite{AAABN09}.

Let us now analyze this algorithm. We say that a request $e_i$ is \emph{small} if it can be covered by a link of cost at most $\OPT/n$, and \emph{large} otherwise. The \emph{incremental cost} incurred in serving a request is the increase in the algorithm's cost when it serves the request. The total incremental cost incurred in serving small requests is at most $\OPT$ since the incremental cost for each small request is at most $\OPT/n$ and there can be at most $n$ requests total (one request per edge of the path). What about the total incremental cost for large requests? The online fractional set cover algorithm of \cite{AAABN09} has a competitive ratio of $O(\log d)$ if each element is contained in at most $d$ sets. Thus, to show that the total incremental cost for large requests is at most $O(\log \log n)\OPT$, it suffices to prove that each large request is covered by at most $2 \log n$ links. Consider a large request $e_i$. Since we had removed links of cost more than $\OPT$ and $e_i$ is large, it is covered only by links of cost between $\OPT/n$ and $\OPT$. Moreover, since costs are powers of $2$, by Property~\ref{def:minimal:class} of minimal instances, $e_i$ is covered by at most $\log n$ links.

We now describe our online fractional path augmentation algorithm that does not need to know the value of $\OPT$ upfront. Let $\OPT_i$ denote the cost of the optimal integral solution for the first $i$ requests $e_1, \ldots, e_i$; Lemma~\ref{lem:LP-integral} implies that $\OPT_i$ can be found by computing a basic feasible solution of LP~\eqref{lp:rP}. For each request $e_i$, if $e_i$ can be covered by a link $\ell$ of cost at most $\OPT_i/n$, then we set the variable $x(\ell) = 1$; otherwise, we perform multiplicative update as follows. 
Let $L_i = \{\ell : e_i \in P(\ell) \wedge c(\ell) \in [\OPT_i/n, 2\OPT_i]\}$. The multiplicative update step does the following: while $\sum_{\ell \in L_i}x_{\ell} < 1$, for each $\ell \in L_i$, increase $x(\ell)$ continuously at the rate
\[\frac{dx(\ell)}{dt} = \frac{x(\ell) + 1/\log n}{c(\ell)}.\]
As before, we say that request $e_i$ is \emph{small} if there exists a link of cost at most $\OPT_i/n$ covering it, and \emph{large} otherwise.

\begin{lemma}
  \label{lem:fractional}
  The solution $x$ computed by the online fractional path augmentation algorithm above has cost $O(\log \log n)\OPT$.
\end{lemma}

\begin{proof}
  It is easy to see that again, the total incremental cost incurred in serving small requests is at most $\OPT$. However, bounding the total incremental cost incurred in serving large requests requires more work.
  Say that an integral solution $F$ is \emph{restricted} if for every large request $e_i$, there is a link in $F$ covering $e_i$ that is also in $L_i$, i.e. $F$ is feasible even if $e_i$ is restricted to be covered only by links in $L_i$. Since costs are powers of $2$, Property~\ref{def:minimal:class} of minimal instances implies that $|L_i| \leq O(\log n)$ and so, as before, the analysis in \cite{AAABN09} implies that the total incremental cost incurred in serving large requests is at most $O(\log \log n) c(F)$ for every restricted solution $F$. It remains to show that there exists a restricted solution $F$ with cost at most $O(\OPT)$.

  Define phase $j$ of the algorithm to consist of the time steps $i$ for which $\OPT_i \in [2^j, 2^{j+1})$ and $i(j)$ to be the last time step of phase $j$. Let $\OPT_j = \OPT_{i(j)}$, the value of an optimal integral solution $F^*_j$ for the requests seen by the end of phase $j$. For each large request $e_i$ in phase $j$, since $\OPT_j < 2 \OPT_i$ and $F^*_j$ is integral, it must be the case that $e_i$ is covered by $F^*_j$ using a link in $L_i$. Let $\hat{F}_j$ be the subset of $F^*_j$ that covers the large requests of phase $j$. Therefore, we get that $F = \bigcup_j \hat{F}_j$ is a restricted solution. The cost of $F$ is $c(F) \leq \sum_{j\leq m} 2^{j+1} \leq 2^{m+2}$, where $m$ is the final phase. Since $\OPT = \OPT_m \geq 2^m$, we have that $c(F) \leq 4\OPT$. This proves that there exists a restricted solution $F$ with cost at most $O(\OPT)$. Therefore, we conclude that $x$ has cost at most $O(\log \log n)\OPT$.
\end{proof}

This proves Theorem~\ref{thm:rand-path-ub}.

\bibliographystyle{alpha}
{\small \bibliography{tap}}

\end{document}